\documentclass{article}
\usepackage[utf8]{inputenc}
\usepackage{booktabs}
\usepackage{amsmath}
\usepackage{amssymb}
\usepackage{amsthm}
\usepackage{mathtools}
\usepackage{microtype}
\usepackage{paralist}
\usepackage{tikz}
\tikzstyle{citation}=[->,shorten >=1pt]
\tikzstyle{mergeable}=[very thick]

\usepackage[sort,numbers]{natbib}
\usepackage[colorlinks=true,linkcolor=red!50!black,citecolor=green!50!black,pagebackref]{hyperref}
\usepackage[textsize=footnotesize,textwidth=1.5cm]{todonotes}
\usepackage[ruled,vlined,linesnumbered]{algorithm2e}
\usepackage{colortbl}
\usepackage[capitalize,noabbrev,nameinlink]{cleveref}
\usepackage{authblk}
\usepackage{subcaption}

\renewcommand*{\backref}[1]{}
\renewcommand*{\backrefalt}[4]{%
\ifcase #1%
\marginpar{\tiny no cite}
\or
 Cited on p.~#2.%
\else
  Cited on pp.~#2.%
\fi
}

\newcommand{\true}{\texttt{true}}
\newcommand{\false}{\texttt{false}}

\newcommand{\dsaia}{\textsf{ai10-2011}}
\newcommand{\dsaib}{\textsf{ai10-2013}}
\newcommand{\dsijcai}{\textsf{ijcai-2013}}

\pgfdeclarelayer{background}
\pgfsetlayers{background,main}

\newcommand{\mmgreedymax}{\textsf{GreedyMax}}
\newcommand{\mmgreedymin}{\textsf{GreedyMin}}
\newcommand{\mmmax}{\textsf{Maximum}}
\newcommand{\mmramsey}{\textsf{Ramsey}}

\newcommand{\decprob}[3]{%
  \par\medskip
  \noindent \textsc{#1}
  \par\noindent\hangindent=\parindent \textbf{Question:}  #3
  \par  \medskip
}

\newcommand\degi{\text{deg}^{\text{in}}}

\newcommand\Ni{N^{\text{in}}}

\newcommand{\hind}{h\nobreakdash-index}
\newcommand{\hinds}{h\nobreakdash-indices}
\newcommand{\Hind}{h\nobreakdash-Index}

\newcommand{\hindM}{\textsc{\hind{} Manipulation}}

\newcommand{\hindS}{\textsc{Dividing}}
\newcommand{\hindE}{\textsc{Ex\-trac\-ting}}
\newcommand{\hindA}{\textsc{Atomizing}}
\newcommand{\hindkS}{\textsc{Cautious \hindS{}}}
\newcommand{\hindkE}{\textsc{Cautious \hindE{}}}
\newcommand{\hindlE}{\textsc{Con\-ser\-va\-tive \hindE{}}}
\newcommand{\hindlS}{\textsc{Con\-ser\-va\-tive \hindS{}}}
\newcommand{\hindlA}{\textsc{Con\-ser\-va\-tive \hindA{}}}

\newcommand{\merge}{\ensuremath{\mathcal P}}
\newcommand{\mergelt}{\ensuremath{P}}

\newcommand{\unmerge}{\ensuremath{\mathcal R}}
\newcommand{\unmergelt}{\ensuremath{R}}

\newtheorem{theorem}{Theorem}
\theoremstyle{definition}

\newtheorem{corollary}{Corollary}

\newtheorem{proposition}{Proposition}

\DeclareMathOperator\scites{sumCite}
\DeclareMathOperator\ucites{unionCite}
\DeclareMathOperator\mcites{fusionCite}
\DeclareMathOperator*{\argmax}{arg\,max}

\makeatletter
\def\NAT@spacechar{~}
\makeatother

\title{h-Index Manipulation by Undoing Merges\thanks{An extended abstract of this article appeared in the proceedings of the 22nd European Conference on Artificial Intelligence (ECAI~'16)~\citep{BKM+16}.
This full version contains additional and corrected experimental results and strengthened hardness results (\cref{thm:hindANP}). The following errors in the previously performed experiments were corrected: 
(1) The algorithm (\mmramsey) for generating initially merged articles was previously not described accurately. The description is now more accurate and we consider additional algorithms to avoid bias in the generated instances. (2) Two authors from the \dsaia\ and \dsaib\ data sets with incomplete data have been used in the experiments; these authors are now omitted. (3) There were several technical errors in the code relating to the treatment of article and cluster identifiers of the crawled articles. This led to inconsistent instances and thus erroneous possible h-index increases. All these errors have been corrected.}}

\author[1,2]{René van Bevern}

\affil[1]{Mathematical Center in Akademgorodok, Novosibirsk, Russian Federation}

\affil[2]{Department of Mechanics and Mathematics, Novosibirsk State University, Novosibirsk, Russian Federation, \texttt{rvb@nsu.ru}}

\author[3]{Christian Komusiewicz}

\affil[3]{Fachbereich Mathematik und Informatik, Philipps-Universit\"at Marburg, Marburg, Germany, \texttt{komusiewicz@informatik.uni-marburg.de}}

\author[4]{Hendrik~Molter}
\author[4]{Rolf~Niedermeier}

\affil[4]{Algorithmics and Computational Complexity, Fakult\"at IV, 
 TU Berlin, Germany, \texttt{\{h.molter,~rolf.niedermeier,~toby.walsh\}@tu-berlin.de}}

\author[5]{Manuel~Sorge}

\affil[5]{Institute of Informatics, University of Warsaw, Poland, \texttt{manuel.sorge@mimuw.edu.pl}}

\author[4]{Toby Walsh}

\begin{document}
\pagestyle{plain}
\thispagestyle{plain}
\maketitle

\begin{abstract}
The \hind{} is an important bibliographic 
measure used to assess the performance of researchers. 
Dutiful researchers merge different versions of their articles in their Google Scholar profile even though this can decrease their \hind{}. In this article, we study the manipulation of the \hind{} by undoing such merges. In contrast to manipulation by merging articles (van Bevern et al.~[\emph{Artif.\ Intel.} 240:19--35, 2016]) such manipulation is harder to detect. 
We present numerous results on computational
complexity (from linear-time algorithms to parameterized 
computational hardness results) and empirically indicate that at least small 
improvements of the \hind{} by splitting merged articles are unfortunately easily 
achievable.

\bigskip
\noindent\emph{Keywords:} Google scholar profiles, citation graph,  article splitting, NP-hard problems, parameterized complexity, experimental algorithmics

\end{abstract}

\section{Introduction}

We suppose that an author has a publication profile, for example in Google Scholar, that consists of single articles and aims to increase her or his \hind{}\footnote{The \hind{} of a
researcher is the maximum number~$h$ such that he or she has at least~$h$~articles each cited at least~$h$~times~\citep{Hir05}.} by merging articles. This will result in a new article with a potentially higher number of citations. The merging option is provided by Google Scholar to identify different versions of the same article, for example a journal version and its archived version.
 
 Our main points of reference are three publications dealing with the manipulation of the \hind{}, particularly motivated by Google Scholar author profile manipulation~\citep{BKNSW16,KA13,EP16}. Indeed, we will closely follow the notation and concepts introduced by \citet{BKNSW16} and we refer to this work for discussion of related work concerning strategic self-citations to manipulate the \hind~\citep{BK11,LRT14,Vin13}, other citation indices~\citep{Egghe2006,EP16,Woe08b}, and manipulation in general~\citep{fhhcacm10,fpaimag10,oravec2017manipulation}.  The main difference between this work and previous publications is that they focus on \emph{merging} articles for increasing the \hind{}~\citep{BKNSW16,KA13,EP16,BK14} or other indices like the g-index and the i10-index~\citep{EP16}, while we focus on \emph{splitting}.

 In the case of splitting, we assume that, most of the time, an author will
maintain a correct profile in which all necessary merges are performed. Some of
these merges may decrease the \hind{}. For instance, this can be the case when
the two most cited papers are the conference and archived version of the same
article. A very realistic scenario is that at certain times, for example when being evaluated by their dean\footnote{ \citet{Les15} pointed out that the \hind{} is the modern 
equivalent of the old saying ``Deans can't read,
they can only count.'' He also remarked that the idea of
``least publishable units'' by dividing one's reports into multiple
(short) papers has been around since the 1970s.}, an author may temporally undo some of these merges to increase artificially her or his \hind{}.
A further point which distinguishes manipulation by splitting from manipulation by merging is that for merging it is easier to detect whether someone cheats too much. This can be done by looking at the titles of merged articles~\citep{BKNSW16}. In contrast, it is much harder to prove that someone is manipulating by splitting; the manipulator can always claim to be too busy or that he or she does not know how to operate the profile.

The main theoretical conclusion from our work is that \hind{} manipulation by splitting merged articles\footnote{Google Scholar allows authors to group different versions of an article. We call the resulting grouping a \emph{merged article}.  Google Scholar author profiles typically contain many merged articles, e.g.\ an arXiv version with a conference version and a journal version.} 
is typically computationally easier than manipulation by merging. Hence, undoing all merges and then merging from scratch might be computationally intractable in cases while, on the contrary, computing an optimal splitting is computationally feasible.  The only good news (and, in a way, a recommendation) in this sense is that if one would use the citation measure ``$\mcites$'' as defined by \citet{BKNSW16}, then manipulation is computationally much harder than for the ``$\ucites$'' measure used by Google Scholar.
In the practical part of our work, we experimented with data from Google Scholar profiles~\citep{BKNSW16}. 

\paragraph{Models for splitting articles.} 

We consider the publication profile of an author and denote the articles in this profile by~$W\subseteq V$, where $V$~is the set of all articles. Following previous work~\citep{BKNSW16}, we call these articles \emph{atomic}. Merging articles yields a partition~$\merge$ of~$W$ in which each part~$P \in \merge$ with~$|P| \geq 2$ is a \emph{merged article}.

Given a partition~$\merge$ of~$W$, the aim of splitting merged
articles is to find a refined partition~$\unmerge$ of~$\merge$ with a
larger \hind{}, where the \emph{\hind{} of a partition~$\merge$} is the
largest number~$h$ such that there are at least~$h$~parts~$P\in\merge$
whose number~$\mu(P)$ of citations is at least~$h$. Herein, we have
multiple possibilities of defining the number~$\mu(P)$ of citations of
an article in~$\merge$~\citep{BKNSW16}. The first one, $\scites(P)$,
was introduced by \citet{KA13}, and is simply the
sum of the citations of each atomic article in~$P$. Subsequently, \citet{BKNSW16} introduced the  citation
measures~$\ucites$ (used by Google Scholar), where we take the
cardinality of the union of the citations, and~$\mcites$, where we
additionally remove self-citations of merged articles as well as duplicate
citations between merged articles. In generic definitions, we denote
these measures by~$\mu$, see \cref{fig:mergevars} for an
illustration 
and \cref{sec:prelim} for the formal definitions. Note that, to
compute these citation measures, we need a \emph{citation graph}, a
directed graph whose vertices represent articles and in which an arc
from a vertex~$u$ to a vertex~$v$ means that article~$u$ cites
article~$v$.

\begin{figure}[t]
  \hfill
  \begin{subfigure}{2cm}
    \begin{tikzpicture}[y=0.6cm]
      \tikzstyle{edge} = [color=black,opacity=.2,line cap=round, line
      join=round, line width=19pt]

      \tikzstyle{vertex}=[circle,draw,fill=white,minimum
      size=10pt,inner sep=1pt,font=\footnotesize]

      \node[vertex] (1) at (0,0) {}; \node[vertex] (2) at (0,1) {};
      \node[vertex] (3) at (0,2) {};

      \node[vertex] (4) at (1,0) {1}; \node[vertex] (5) at (1,1) {2};
      \node[vertex] (6) at (1,2) {2};

      \draw[citation] (1)--(4); \draw[citation] (1)--(5);
      \draw[citation] (2)--(6); \draw[citation] (3)--(6);
      \draw[citation] (4)--(5);
      \begin{pgfonlayer}{background}
        \draw[edge] (2.center)--(3.center); \draw[edge]
        (4.center)--(5.center);
      \end{pgfonlayer}
    \end{tikzpicture}
    \caption{Unmerged}
    \label{fig:unmerged}
  \end{subfigure}
  \hfill
  \begin{subfigure}{2cm}
    \begin{tikzpicture}[y=0.6cm]
      \tikzstyle{edge} = [color=black,opacity=.2,line cap=round,
      line join=round, line width=19pt]
      
      \tikzstyle{vertex}=[circle,draw,fill=white,minimum
      size=10pt,inner sep=1pt,font=\footnotesize]
      
      \node[vertex] (1) at (0,0) {}; \node[vertex] (3) at (0,2) {};
      
      \node[vertex] (4) at (1,0) {3}; \node[vertex] (6) at (1,2)
      {2};
      
      \draw[citation] (1)to[out=20,in=160](4); \draw[citation]
      (1)to[out=-20,in=-160](4); \draw[citation]
      (3)to[out=20,in=160](6); \draw[citation]
      (3)to[out=-20,in=-160](6);

      \draw[citation, loop above] (4)to (4);
      \begin{pgfonlayer}{background}
        \draw[edge,transparent] (2.center)--(3.center); \draw[edge,transparent]
        (4.center)--(5.center);
      \end{pgfonlayer}
    \end{tikzpicture}
    \caption{$\scites$}\label{fig:scites}
  \end{subfigure}
  \hfill
  \begin{subfigure}{2cm}
    \begin{tikzpicture}[y=0.6cm]
      \tikzstyle{edge} = [color=black,opacity=.2,line cap=round, line
      join=round, line width=19pt]
      
      \tikzstyle{vertex}=[circle,draw,fill=white,minimum size=10pt,inner
      sep=1pt,font=\footnotesize]
      
      \node[vertex] (1) at (0,0) {}; \node[vertex] (3) at (0,2) {};
      
      \node[vertex] (4) at (1,0) {2}; \node[vertex] (6) at (1,2) {2};
      
      \draw[citation, loop above] (4)to (4);

      \draw[citation] (1)--(4); \draw[citation] (3)to[out=20,in=160](6);
      \draw[citation] (3)to[out=-20,in=-160](6);
      \begin{pgfonlayer}{background}
        \draw[edge,transparent] (2.center)--(3.center); \draw[edge,transparent]
        (4.center)--(5.center);
      \end{pgfonlayer}
    \end{tikzpicture}
    \caption{$\ucites$}
    \label{fig:ucites}
  \end{subfigure}
  \hfill
  \begin{subfigure}{1.9cm}
    \begin{tikzpicture}[y=0.6cm]
      \tikzstyle{edge} = [color=black,opacity=.2,line cap=round,
      line join=round, line width=19pt]
      
      \tikzstyle{vertex}=[circle,draw,fill=white,minimum
      size=10pt,inner sep=1pt,font=\footnotesize]
      
      \node[vertex] (1) at (0,0) {}; \node[vertex] (3) at (0,2) {};
      
      \node[vertex] (4) at (1,0) {1}; \node[vertex] (6) at (1,2)
      {1};
      
      \draw[citation] (1)--(4); \draw[citation] (3)--(6);
      \begin{pgfonlayer}{background}
        \draw[edge,transparent] (2.center)--(3.center); \draw[edge,transparent]
        (4.center)--(5.center);
      \end{pgfonlayer}
    \end{tikzpicture}
    \caption{$\mcites$}
    \label{fig:mcites}
  \end{subfigure}
  \hfill\mbox{}
  \caption{Vertices represent articles, arrows represent citations,
    numbers are citation counts. The articles on a gray
    background in \subref{fig:unmerged} have been merged in
    \subref{fig:scites}--\subref{fig:mcites}, and citation
    counts are given according to the measures~$\scites$,~$\ucites$,
    and~$\mcites$, respectively. The arrows represent the citations counted by the corresponding measure.}
  \label{fig:mergevars}
\end{figure}

In this work, we introduce three different operations that may be used for 
undoing merges in a merged article~$a$:
\smallskip
  \begin{description}
  \item[Atomizing:] splitting $a$ into all its atomic
    articles,
  \item[Extracting:] splitting off a single atomic
    article from~$a$, and
  \item[Dividing:] splitting $a$ into two parts arbitrarily.
  \end{description}
\smallskip
See \cref{fig:splitvars} for an illustration of the three splitting operations. 
Note that the atomizing, extracting, and dividing operations are successively strictly more powerful in the sense that successively larger \hinds{} can be achieved. 
Google Scholar offers the extraction operation. Multiple applications of the extraction operation can, however, simulate atomizing and dividing.
 
\begin{figure}[t]
  \centering
  \hfill
  \begin{subfigure}{2cm}
    \centering
    \begin{tikzpicture}[y=0.8cm]
      \tikzstyle{spacing} = [color=black,opacity=0,line cap=round, line
      join=round, line width=19pt]
      \tikzstyle{edge} = [color=black,opacity=.2,line cap=round, line
      join=round, line width=19pt]
      \tikzstyle{vertex}=[circle,draw,fill=white,minimum
      size=10pt,inner sep=1pt,font=\footnotesize]

      \node[vertex] (1) at (0,0) {}; \node[vertex] (2) at (0,1) {};
      \node[vertex] (3) at (0,2) {}; \node[vertex] (4) at (0,3) {};

      \node[vertex] (5) at (1,0) {}; \node[vertex] (6) at (1,1) {};
      \node[vertex] (7) at (1,2) {}; \node[vertex] (8) at (1,3) {};

      \draw[citation] (1)--(5); \draw[citation] (1)--(6);
      \draw[citation] (2)--(7); \draw[citation] (3)--(8);
      \draw[citation] (4)--(8);
      \begin{pgfonlayer}{background}
         \draw[edge]
        (5.center)--(8.center);
      \end{pgfonlayer}
      \begin{pgfonlayer}{background}
         \draw[spacing]
        (5.center)--(8.center);
      \end{pgfonlayer}
    \end{tikzpicture}
    \caption{Merged}
    \label{fig:merged}
  \end{subfigure}
  \hfill
  \begin{subfigure}{1.8cm}
    \centering
    \begin{tikzpicture}[y=0.8cm]
          \tikzstyle{spacing} = [color=black,opacity=0,line cap=round, line
      join=round, line width=19pt]
      \tikzstyle{edge} = [color=black,opacity=.2,line cap=round, line
      join=round, line width=19pt]

      \tikzstyle{vertex}=[circle,draw,fill=white,minimum
      size=10pt,inner sep=1pt,font=\footnotesize]

      \node[vertex] (1) at (0,0) {}; \node[vertex] (2) at (0,1) {};
      \node[vertex] (3) at (0,2) {}; \node[vertex] (4) at (0,3) {};

      \node[vertex] (5) at (1,0) {1}; \node[vertex] (6) at (1,1) {1};
      \node[vertex] (7) at (1,2) {1}; \node[vertex] (8) at (1,3) {2};

      \draw[citation] (1)--(5); \draw[citation] (1)--(6);
      \draw[citation] (2)--(7); \draw[citation] (3)--(8);
      \draw[citation] (4)--(8);
      \begin{pgfonlayer}{background}
        \draw[edge,color=white] (5.center)--(7.center);
      \end{pgfonlayer}
            \begin{pgfonlayer}{background}
         \draw[spacing]
        (5.center)--(8.center);
      \end{pgfonlayer}
    \end{tikzpicture}
    \caption{Atomizing}\label{fig:atomize}
  \end{subfigure}\hfill
  \begin{subfigure}{1.8cm}
    \centering
    \begin{tikzpicture}[y=0.8cm]
          \tikzstyle{spacing} = [color=black,opacity=0,line cap=round, line
      join=round, line width=19pt]
      \tikzstyle{edge} = [color=black,opacity=.2,line cap=round, line
      join=round, line width=19pt]

      \tikzstyle{vertex}=[circle,draw,fill=white,minimum
      size=10pt,inner sep=1pt,font=\footnotesize]

      \node[vertex] (1) at (0,0) {}; \node[vertex] (2) at (0,1) {};
      \node[vertex] (3) at (0,2) {}; \node[vertex] (4) at (0,3) {};

      \node[vertex] (5) at (1,0) {}; \node[vertex] (6) at (1,1) {};
      \node[vertex] (7) at (1,2) {}; \node[vertex] (8) at (1,3) {2};

      \draw[citation] (1)--(5); \draw[citation] (1)--(6);
      \draw[citation] (2)--(7); \draw[citation] (3)--(8);
      \draw[citation] (4)--(8);
      \begin{pgfonlayer}{background}
         \draw[edge]
        (5.center)--(7.center);
      \end{pgfonlayer}
            \begin{pgfonlayer}{background}
         \draw[spacing]
        (5.center)--(8.center);
      \end{pgfonlayer}
    \end{tikzpicture}
    \caption{Extracting}
    \label{fig:extract}
  \end{subfigure}\hfill
  \begin{subfigure}{1.9cm}
    \centering
    \begin{tikzpicture}[y=0.8cm]
          \tikzstyle{spacing} = [color=black,opacity=0,line cap=round, line
      join=round, line width=19pt]
      \tikzstyle{edge} = [color=black,opacity=.2,line cap=round, line
      join=round, line width=19pt]

      \tikzstyle{vertex}=[circle,draw,fill=white,minimum
      size=10pt,inner sep=1pt,font=\footnotesize]

      \tikzstyle{citation}=[->] \tikzstyle{mergeable}=[very thick]
      
      \node[vertex] (1) at (0,0) {}; \node[vertex] (2) at (0,1) {};
      \node[vertex] (3) at (0,2) {}; \node[vertex] (4) at (0,3) {};

      \node[vertex] (5) at (1,0) {}; \node[vertex] (6) at (1,1) {};
      \node[vertex] (7) at (1,2) {}; \node[vertex] (8) at (1,3) {};

      \draw[citation] (1)--(5); \draw[citation] (1)--(6);
      \draw[citation] (2)--(7); \draw[citation] (3)--(8);
      \draw[citation] (4)--(8);
      \begin{pgfonlayer}{background}
         \draw[edge] (5.center)--(6.center);
        \draw[edge] (7.center)--(8.center);
      \end{pgfonlayer}
            \begin{pgfonlayer}{background}
         \draw[spacing]
        (5.center)--(8.center);
      \end{pgfonlayer}
    \end{tikzpicture}    \caption{Dividing}
    \label{fig:split}
  \end{subfigure}
  \hfill\mbox{}
  \caption{Vertices represent articles, arrows represent citations,
    numbers are citation counts. The articles on a gray background
    have been merged in the initial profile \subref{fig:merged} and
    correspond to remaining merged articles after applying one operation in
    \subref{fig:extract} and \subref{fig:split}. Each (merged) article has the same
    citation count, regardless of the used
    measure~$\scites$,~$\ucites$, and~$\mcites$. }
  \label{fig:splitvars}
\end{figure}

The three splitting operations lead to three 
problem variants, each taking as input a citation
graph~$D=(V,A)$, a set $W\subseteq V$ of articles belonging to the author, a partition~$\merge$ of~$W$ that defines already-merged
articles, and a non-negative integer~$h$ denoting the \hind{} to achieve.
For $\mu \in \{\scites,\allowbreak
\ucites, \mcites\}$, we define the following problems.
  \vspace{-2pt}
\decprob{\hindA{}($\mu$)}
{$D=(V,A)$, $W\subseteq V$, $\merge$ of~$W$, and~$h$.}
  {Is there a partition~$\unmerge$ of~$W$ such that
  \begin{compactenum}[i)]
  \item for each~$\unmergelt \in \unmerge$ either~$|\unmergelt|=1$ or there
  is a~$\mergelt \in \merge$ such that~$\unmergelt=\mergelt$,
  \item the \hind{} of~$\unmerge$ is at least~$h$ with respect to~$\mu$?
  \end{compactenum}}
  \vspace{-9pt}
\decprob{\hindE{}$(\mu)$}
{$D=(V,A)$, $W\subseteq V$, $\merge$ of~$W$, and~$h$.}
{Is there a partition~${\unmerge}$ of~$W$ such that
  \begin{compactenum}[i)]
  \item for each~$\unmergelt \in \unmerge$ there is a~$\mergelt \in
  \merge$ such that~$\unmergelt \subseteq \mergelt$,
  \item for each~$\mergelt \in \merge$ we have~$|\{\unmergelt \in \unmerge
  \ | \ \unmergelt \subset \mergelt \ \text{and} \ |\unmergelt| > 1\}| \leq 1$,
  \item the \hind{} of~$\unmerge$ is at least~$h$ with respect to~$\mu$?
  \end{compactenum}}
   \vspace{-9pt}
\decprob{\hindS{}$(\mu)$}
{$D=(V,A)$, $W\subseteq V$, $\merge$ of~$W$, and~$h$.}
{Is there a partition~$\unmerge$ of~$W$ such that
  \begin{compactenum}[i)]
  \item for each~$\unmergelt \in \unmerge$ there is a~$\mergelt \in \merge$
  such that~$\unmergelt \subseteq \mergelt$,
  \item the \hind{} of~$\unmerge$ is at least~$h$ with respect to~$\mu$?
  \end{compactenum}}

\paragraph{Conservative splitting.}
We study for each of the problem variants an additional upper bound
on the number of merged articles that are split. We call these
variants \emph{conservative}: if an insincere author would like to
manipulate his or her profile temporarily, then he or she might prefer a
manipulation that can be easily undone. To formally define \hindlA{}, \hindlE{}, and \hindlS{}, we add the
following restriction to the partition~$\unmerge$: ``the
number~$|\merge \setminus \unmerge|$ of changed articles is at
most~$k$''.

A further motivation for the conservative variants is that, in a Google
Scholar profile, an author can click on a merged article and tick a
box for each atomic article that he or she wants to extract. Since
Google Scholar uses the $\ucites$ measure~\citep{BKNSW16}, \hindlE$(\ucites)$ thus corresponds closely to manipulating the Google Scholar \hind{} via few of the splitting operations available to the user.

 \paragraph{Cautious splitting.} For each splitting operation, we also
study an upper bound~$k$ on the number of split operations.  Following our previous work~\citep{BKNSW16}, we call this variant 
\emph{cautious}. In the case of atomizing, conservativity and
caution coincide since exactly one operation is performed per
changed article. Thus, we obtain two cautious problem variants:
\hindkE{} and \hindkS{}.  For both we add the following restriction to
the partition $\unmerge$: ``the number~$|\unmerge|-|\merge|$ of
extractions (or divisions, respectively) is at most~$k$''. In both variants we consider $k$ to be part of the input.

\paragraph{Our results.}
We investigate the parameterized computational complexity of our problem variants with respect to the parameters ``the h-index $h$ to achieve'', and in the conservative case ``the number $k$ of modified merged articles'', and in the cautious case ``the number $k$ of splitting operations''. 
To put it briefly, the goal is to exploit potentially small parameter values 
(that is, special properties of the input instances) in order to gain 
efficient algorithms for problems that are in general computationally hard.
In our context, the choice of the parameter~$h$ is motivated by the scenario that young researchers may have an incentive to increase their h-index and since they are young, the h-index $h$ to achieve is not very large. The conservative and cautious scenario tries to capture that the manipulation can easily be undone or is hard to detect, respectively. Hence, it is well motivated that the parameter $k$ shall be small.
Our theoretical (computational complexity classification) results are summarized in \cref{tab:results} (see Section~\ref{sec:prelim} for further definitions). The measures $\scites$ and $\ucites$ behave basically the same. In particular, in case of atomizing and extracting, manipulation is doable in linear time, while $\mcites$ mostly leads to (parameterized) intractability, that is, to high worst-case computational complexity. Moreover, the dividing operation (the most general one) seems to lead to computationally much harder problems than atomizing and extracting.

We performed experiments with real-world data~\citep{BKNSW16} and the
mentioned linear-time algorithms, in particular for the case directly relevant to Google
Scholar, that is, using the extraction operation and the $\ucites$
measure. Our general findings are that increases of the \hind{} by one
or two typically are easily achievable with few operations. 
The good news is that dramatic manipulation opportunities due to splitting are 
rare. They cannot be excluded, however, and they 
could be easily executed when relying
on standard operations and measures (as used in Google Scholar). 
Working with $\mcites$ instead of the other two could substantially 
hamper manipulation.

\begin{table}
    
    \caption{Computational (time) complexity of the various variants of manipulating 
    the \hind{}
    by splitting operations (see \cref{sec:prelim} for definitions). For all FPT and W[1]-hardness results we also
    show NP-hardness. 
    \\$\dagger$:
    wrt.\ parameter $h$, the \hind{} to achieve.\\$\diamond$: wrt.\ parameter $k$, the number of
    operations.\\$\star$: wrt.\ parameter $h + k + s$, where $s$ is the largest number of articles merged into one.\\
    $\ddagger$: NP-hard even if $k=1$ (\cref{prop:hindSNPh}).\\
    $\odot$: Parameterized complexity wrt.\ $h$ open.\\}
\label{tab:results}
  \centering
  \begin{tabular}{ l  l  l  l }
    \toprule
    Problem & $\scites$ / $\ucites$ & $\mcites$ \\ 
    \midrule
    Atomizing & Linear (\cref{thm:hindA}) & FPT$^\dagger$
    (\cref{thm:hindANP,thm:hindAmcites})\\
    Conservative A. & Linear (\cref{thm:hindA}) & W[1]-h$^\star$
    (\cref{thm:hindkAW1})\\
    \midrule
    Extracting & Linear (\cref{thm:hindE}) & NP-h$^\odot$ (\cref{thm:hindANP}) \\ 
    Conservative E. & Linear (\cref{thm:hindE}) &
    W[1]-h$^\star$ (\cref{cor:mcites}) \\
    Cautious E. & Linear (\cref{thm:hindE}) &
    W[1]-h$^\star$ (\cref{cor:mcites}) \\
    \midrule
    Dividing & FPT$^\dagger$
    (\cref{thm:hindS}) &  NP-h$^\odot$ (\cref{prop:hindSNPh}) \\
    Conservative D. & FPT$^{\dagger,\ddagger}$
    (\cref{thm:hindS}) & W[1]-h$^\star$ (\cref{cor:mcites})
    \\
        Cautious D. & W[1]-h$^{\diamond,\odot}$
    (\cref{thm:hindkS}) & W[1]-h$^\star$ (\cref{cor:mcites})
    \\
    \bottomrule
    \end{tabular}
\end{table}

\section{Preliminaries}\label{sec:prelim}

Throughout this work, we use~$n:=|V|$ for the number of input
articles and~$m:=|A|$ for the overall number of arcs in the input citation graph~$D=(V,E)$.
Let $\degi(v)$~denote the indegree of an article~$v$ in a citation
graph~$D=(V,A)$, that is, $v$'s number of citations. Furthermore, let $\Ni_D(v) :=
\{u \mid (u, v) \in A\}$ denote the set of articles that cite~$v$ and
$\Ni_{D-W}(v) := \{u \mid (u, v) \in A \wedge u \notin W\}$ be the set
of articles outside~$W$ that cite~$v$. 
For each part~$\mergelt\in\merge$, the
following three measures for the number~$\mu(P)$ of
citations of~$P$ have been introduced~\citep{BKNSW16}. 
They are illustrated in \cref{fig:mergevars}.
The measure

\begin{align*}
  \scites(\mergelt)&:=\sum_{v\in \mergelt}\degi(v)\\
  \intertext{defines the number of citations of a merged article~$P$
    as the sum of the citations of the atomic articles it contains.
    This measure was proposed by \citet{KA13}. In contrast,
    the measure}
  \ucites(\mergelt)&:=\Bigl|\bigcup_{v\in \mergelt}\Ni_D(v)\Bigr|\\
  \intertext{defines the number of citations of a merged article~$P$
    as the number of distinct atomic articles citing at least one
    atomic article in~$P$.  
    Google Scholar uses the~$\ucites$~measure~\citep{BKNSW16}.
    The measure} \mcites(\mergelt)&:=\Bigl|\bigcup_{v\in
    \mergelt}\Ni_{D-W}(v)\Bigr| +{}\\ &\hspace{-3em}
  \smashoperator[l]{\sum_{\substack{\mergelt'\in\,\merge\setminus\, \{\mergelt}\}}}
  \begin{cases}
    1&\text{if }\exists v\in\mergelt' \exists w\in\mergelt:(v,w)\in A,\\
    0&\text{otherwise}
  \end{cases}
\end{align*}

\noindent is perhaps the most natural one: at most one citation of a part~$P'\in\merge$ to a part~$P\in\merge$ is counted, that is, we
additionally remove duplicate citations between merged articles and
self-citations of merged articles.

 Our theoretical analysis is in the framework of parameterized
complexity~\citep{Cy15,DF13,FG06,Nie06}. That is, for those problems
that are NP-hard, we study the influence of a \emph{parameter}, an
integer associated with the input, on the computational
complexity. For a problem~$P$, we seek to decide~$P$ using a \emph{fixed-parameter algorithm},
an algorithm with
running time~$f(p) \cdot |q|^{O(1)}$, where $q$ is the input and $f(p)$ a
computable function depending only on the parameter~$p$. If such an
algorithm exists, then $P$~is
\emph{fixed-parameter tractable} (FPT) with respect to~$p$. 
W[1]-hard parameterized problems presumably do not admit FPT algorithms. For instance, to find an order-$k$ clique in an undirected graph is known to be W[1]-hard for the parameter~$k$. W[1]-hardness of a problem~$P$ parameterized by~$p$ can be shown via a \emph{parameterized reduction} from a known W[1]-hard problem~$Q$ parameterized by~$q$. That is, a reduction that runs in~$f(q)\cdot n^{O(1)}$ time on input of size~$n$ with parameter~$q$ and produces instances that satisfy~$p \leq f(q)$ for some function~$f$.

\section{sumCite and unionCite}\label{sec:sucites}
In this section, we study the $\scites$ and $\ucites$ measures. We provide
linear-time algorithms for atomizing and extracting and analyze the parameterized complexity of dividing with respect to the number~$k$ of splits and the \hind{}~$h$ to achieve. In our results for $\scites$ and $\ucites$, we often tacitly use the observation
that local changes to the merged articles do not influence the citations of
other merged articles.

\paragraph{Manipulation by atomizing.}
\SetKwFunction{Atomize}{Atomize}

Recall that the atomizing operation splits a merged article into singletons and that, for the atomizing operation, the notions of \emph{conservative} (touching few articles) and \emph{cautious} (making few split operations) manipulation coincide and are thus both captured by \hindlA{}. 
Both \hindA{} and \hindlA{} are solvable in linear time. Intuitively, it suffices to find the merged articles which, when atomized, increase the number of articles with at least~$h$ citations the most.  This leads to \cref{alg:hindA,alg:hindkA} for \hindA{} and \hindlA{}. Herein, the \Atomize{} operation takes a set~$S$ as input and returns~$\{\{s\} \mid s \in S\}$. The algorithms yield the following theorem.

\begin{algorithm}[t]
  \DontPrintSemicolon
  
  \SetKwFunction{Atomize}{Atomize}
  
  \KwIn{A citation graph~$D=(V,A)$, a set
  $W\subseteq V$ of articles, a partition~$\merge$ of~$W$, a
  nonnegative integer~$h$ and a measure~$\mu$.}
  
  \KwOut{A partition~$\unmerge$ of~$W$.}
  
  \BlankLine
  
  $\unmerge \leftarrow \emptyset$\; 
  \ForEach{$\mergelt \in \merge$}
  {
	$\mathcal A \leftarrow \Atomize{\mergelt}$\;
	\lIf{$\exists A \in \mathcal A \colon \mu(A) \geq h$}
		{
		$\unmerge \leftarrow \unmerge \cup \mathcal A$
		}
	\lElse
		{
		$\unmerge \leftarrow \unmerge \cup \{ \mergelt \}$
		}
  }
  \KwRet{$\unmerge$}
  \caption{Atomizing}
  \label{alg:hindA}
\end{algorithm}

\begin{algorithm}[t]
  \DontPrintSemicolon
  
  \SetKwFunction{Atomize}{Atomize}
  
  \KwIn{A citation graph~$D=(V,A)$, a set
  $W\subseteq V$ of articles, a partition~$\merge$ of~$W$, nonnegative
  integers~$h$ and~$k$, and a measure~$\mu$.}
  
  \KwOut{A partition~$\unmerge$ of~$W$.}
  
  \BlankLine
  
  $\unmerge \leftarrow \merge$\; 
  \ForEach(\nllabel{lin:recstart}){$\mergelt \in \merge$}
  {%
  	$\ell_{P} \gets 0$\;
	$\mathcal A \leftarrow \Atomize{\mergelt}$\;
        $\ell_{P} \leftarrow \ell_P + |\{A \in \mathcal A \mid \mu(A) \geq h\}|$\;
	\lIf{$\mu(P) \geq h$}{$\ell_{P}\leftarrow\ell_{P}-1$}\nllabel{lin:recfin}
  }%
  \For(\nllabel{lin:atomstart}){$i\leftarrow 1$ \KwTo $k$} {
  	$\mergelt^* \leftarrow \argmax_{\mergelt \in \merge}\{\ell_{\mergelt}\}$\;
  	\If{$\ell_{\mergelt^*} > 0$} {
  		$\mathcal A \leftarrow \Atomize{$\mergelt^*$}$\;
  		$\unmerge \leftarrow (\unmerge \setminus \{\mergelt^*\}) \cup \mathcal A$\;
  	}
  	$\ell_{\mergelt^*} \gets -1$\nllabel{lin:atomfin}\;
  }
  \KwRet{$\unmerge$}
  \caption{Conservative Atomizing}
  \label{alg:hindkA}
\end{algorithm}

\begin{theorem}
  \hindA{}$(\mu)$ and \hindlA{}$(\mu)$ are solvable in linear time
  for~$\mu\in\{\scites,\allowbreak\ucites\}$.\label{thm:hindA}
\end{theorem}

\begin{proof}
  We first consider \hindA{}$(\mu)$.
Let~$\unmerge$ be a partition created from a partition~$\merge$ by atomizing a
part~$\mergelt^* \in \merge$. Observe that for all~$\mergelt \in \merge$
and~$\unmergelt \in \unmerge$ we have that~$\mergelt = \unmergelt$
implies~$\mu(\mergelt) = \mu(\unmergelt)$, for~$\mu\in\{\scites,\ucites\}$.
Intuitively this means that atomizing a single part~$\mergelt^* \in \merge$ does
not alter the~$\mu$-value of any other part of the partition.

\cref{alg:hindA} computes a partition~$\unmerge$ that has a
maximal number of parts~$\unmergelt$ with~$\mu(\unmergelt) \geq h$ that can
be created by applying atomizing operations to~$\merge$: It applies the atomizing operation to each part~$\mergelt \in \merge$ if there is
at least one singleton~$A$ in the atomization of~$\mergelt$ with~$\mu(A) \geq
h$. By the above observation, this cannot decrease the total
number of parts in the partition that have a~$\mu$-value of at least~$h$.
Furthermore, we have that for all~$\unmergelt \in \unmerge$, we cannot
potentially increase the number of parts with~$\mu$-value at least~$h$ by
atomizing~$\unmergelt$. Thus, we get the maximal number of parts~$\unmergelt$
with~$\mu(\unmergelt) \geq h$ that can be created by applying atomizing
operations to~$\merge$.

Obviously, if~$\unmerge$ has at least~$h$ parts~$\unmergelt$
with~$\mu(\unmergelt) \geq h$, we face a yes-instance. Conversely, if the input
is a yes-instance, then there is a number of atomizing operations that can be
applied to~$\merge$ such that the resulting partition~$\unmerge$ has at
least~$h$ parts~$\unmergelt$ with~$\mu(\unmergelt) \geq h$.

It is easy to see that the algorithm runs in linear time and finds a
yes-instance if it exists. If the output partition~$\unmerge$ does not have at
least~$h$ parts~$\unmergelt$ with~$\mu(\unmergelt) \geq h$, then the input is a
no-instance.

\medskip
 The pseudocode for solving \hindlA$(\mu)$ is given in \cref{alg:hindkA}. 
  First, in Lines~\ref{lin:recstart}--\ref{lin:recfin}, for each part~\(\mergelt\), \cref{alg:hindkA} records how many singletons~$A$ with~$\mu(A) \geq h$ are created when atomizing~\(P\). Then, in Lines~\ref{lin:atomstart}--\ref{lin:atomfin}, it repeatedly atomizes the part yielding the most such singletons. This procedure creates the maximum number of parts that have a $\mu$-value of at least~$h$, since the \(\mu\)-value cannot be increased by exchanging one of these atomizing operations by another.

  Obviously, if $\unmerge$~has at least $h$~parts~$\unmergelt$ with $\mu(\unmergelt) \geq h$, then we face a yes-instance. Conversely, if the input is a yes-instance, then there are $k$~atomizing operations that can be applied to $\merge$ to yield an \hind{} of at least~$h$. Since \cref{alg:hindkA} takes successively those operations that yield the most new parts with $h$~citations, the resulting partition~$\unmerge$ has at least $h$~parts~$\unmergelt$ with $\mu(\unmergelt) \geq h$. It is not hard to verify that the algorithm has linear running time.
\end{proof}

\paragraph{Manipulation by extracting.} 
Recall that the extracting operation removes a single article from a
merged article. All variants of the extraction problem are solvable in
linear time. 
Intuitively, in the cautious
case, it suffices to find $k$~extracting operations that each increase
the number of articles with~$h$ citations. In the conservative case,
we determine for each merged article a set of extraction operations
that increases the number of articles with~$h$ citations the
most. Then we use the extraction operations for those~$k$
merged articles that yield the $k$~largest increases in the number of
articles with~$h$ citations.  This leads to \cref{alg:hindE,alg:hindkE,alg:hindlE} for \hindE{}, \hindkE{}, and \hindlE{}, respectively, which yield the following theorem.

\begin{algorithm}[t]
  \DontPrintSemicolon
  
  \KwIn{A citation graph~$D=(V,A)$, a set
  $W\subseteq V$ of articles, a partition~$\merge$ of~$W$, a nonnegative
  integer~$h$ and a measure~$\mu$.}
  
  \KwOut{A partition~$\unmerge$ of~$W$.}
  
  \BlankLine
  
  $\unmerge \leftarrow \emptyset$\;
  \ForEach{$\mergelt \in \merge$}
  {
  	\ForEach{$v \in \mergelt$} {
  	\If{$\mu(\{v\}) \geq h$}
  	{
  	$\unmerge \leftarrow \unmerge \cup \{\{v\}\}$\;
  	$\mergelt \leftarrow \mergelt \setminus \{v\}$\;
  	}
	}
	\lIf{$\mergelt \neq \emptyset$}{$\unmerge \leftarrow \unmerge \cup
  	\{\mergelt\}$} 
  }
  \KwRet{$\unmerge$}
  \caption{Extracting}
  \label{alg:hindE}
\end{algorithm}

\begin{algorithm}[t]
  \DontPrintSemicolon
  
  \KwIn{A citation graph~$D=(V,A)$, a set
  $W\subseteq V$ of articles, a partition~$\merge$ of~$W$, nonnegative
  integers~$h$ and~$k$, and a measure~$\mu$.}
  
  \KwOut{A partition~$\unmerge$ of~$W$.}
  
  \BlankLine
  
  $\unmerge \leftarrow \emptyset$\;
  \ForEach(\nllabel{alg2:formergelt}){$\mergelt \in \merge$}
  {
    \ForEach(\nllabel{alg2:forv}){$v \in \mergelt$} {
      \If(\nllabel{lin:check1}){$k > 0$ and $\mu(\{v\}) \geq h$ and $\mu(\mergelt
  	\setminus \{v\}) \geq h$}%
      {%
        $\unmerge \leftarrow \unmerge \cup \{\{v\}\}$\;
  	$\mergelt \leftarrow \mergelt \setminus \{v\}$\nllabel{alg2:delv}\;
  	$k\leftarrow k-1$\;
      }%
    }
	\lIf{$\mergelt \neq \emptyset$}{$\unmerge \leftarrow \unmerge \cup
  	\{\mergelt\}$} 
  }
  \KwRet{$\unmerge$}
  \caption{Cautious Extracting}
  \label{alg:hindkE}
\end{algorithm}

\begin{algorithm}[t]
  \DontPrintSemicolon
  
  \KwIn{A citation graph~$D=(V,A)$, a set
  $W\subseteq V$ of articles, a partition~$\merge$ of~$W$, nonnegative
  integers~$h$ and~$k$, and a measure~$\mu$.}
  
  \KwOut{A partition~$\unmerge$ of~$W$.}
  
  \BlankLine
  
  \ForEach{$\mergelt \in \merge$}
  {
    	$\ell_\mergelt \gets 0$\;
    	  $\unmerge_\mergelt \leftarrow \emptyset$\;
  	\ForEach{$v \in \mergelt$} {

  	\If(\nllabel{lin:check2}){$\mu(\{v\}) \geq h$ and $\mu(\mergelt \setminus
  	\{v\}) \geq h$} {
  	$\unmerge_\mergelt \leftarrow \unmerge_\mergelt \cup \{\{v\}\}$\;
  	$\mergelt \leftarrow \mergelt \setminus \{v\}$\;
  	$\ell_\mergelt\leftarrow \ell_\mergelt+1$\;
  	}
	}
	\lIf{$\mergelt \neq \emptyset$}{$\unmerge_\mergelt \leftarrow \unmerge_\mergelt
	\cup \{\mergelt\}$} 

  }
  	$\merge^*\gets{}$the $k$~elements of~$\mergelt\in\merge$ with largest
  	$\ell_\mergelt$-values\; \nllabel{lin:count}
  	$\unmerge \leftarrow \bigcup_{\mergelt \in
  	\merge^*}\unmerge_\mergelt \cup (\merge\setminus\merge^*)$\;
  	\KwRet{$\unmerge$}
  \caption{Conservative Extracting}
  \label{alg:hindlE}
\end{algorithm}

\begin{theorem}
  \hindE{}$(\mu)$, \hindlE{}$(\mu)$ and \hindkE{}$(\mu)$ are solvable in linear
  time for~$\mu\in\{\scites,\ucites\}$.\label{thm:hindE}
\end{theorem}

\begin{proof}
  We first consider \hindE{}$(\mu)$. Let~$\unmerge$ be a partition produced from~$\merge$ by extracting an
article from a part~$\mergelt^* \in \merge$. Recall that this does not alter the
$\mu$-value of any other part, i.e., for all $\mergelt \in \merge$ and $\unmergelt \in \unmerge$, we have that $\mergelt = \unmergelt$
implies~$\mu(\mergelt) = \mu(\unmergelt)$ for~$\mu\in\{\scites,\ucites\}$.

Consider \cref{alg:hindE}. It is easy to see that the algorithm
only performs extracting operations and that the running time is polynomial. So
we have to argue that whenever there is a partition~$\unmerge$ that can be
produced by extracting operations from~$\merge$ such that the \hind{} is at
least~$h$, then the algorithm finds a solution.

We show this by arguing that the algorithm produces the maximum number of
articles with at least~$h$ citations possible. Extracting an article that has
strictly less than~$h$ citations cannot produce an \hind{} of at least~$h$
unless we already have an \hind{} of at least~$h$, because the number of
articles with~$h$ or more citations does not increase. Extracting an article
with~$h$ or more citations cannot decrease the number of articles with~$h$ or
more citations. Hence, if there are no articles with at least~$h$ citations that
we can extract, we cannot create more articles with~$h$ or more citations.
Therefore, we have produced the maximum number of articles with~$h$ or more
citations when the algorithm stops.

\medskip   The pseudocode for solving \hindkE$(\mu)$ is given in \cref{alg:hindkE}. We perform up to $k$~extracting operations (\cref{alg2:delv}).  Each of them increases the number of articles that have~$h$ or more citations by one. As \cref{alg:hindkE} checks each atomic article in each merged article, it finds $k$~extraction operations that increase the number of articles with~$h$ or more citations if they exist. Thus, it produces the maximum-possible number of articles that have~$h$ or more citations and that can be created by \(k\)~extracting operations.

  To achieve linear running time, we need to efficiently compute
  \(\mu(P\setminus\{v\})\) in \cref{lin:check1}. This can be done
  by representing articles as integers and using an $n$-element
  array~$A$ which stores throughout the loop in \cref{alg2:forv},
  for each article~\(v\in\Ni_D[P]\), the number~\(A[w]\) of articles
  in~\(P\) that are cited by~\(w\). Using this array, one can compute
  \(\mu(P\setminus\{v\})\) in \(O(\degi(v))\)~time in
  \cref{lin:check1}, amounting to overall linear time. The time needed to maintain array~$A$ is also linear: We initialize it once in the beginning with all zeros. Then, before entering the loop in \cref{alg2:forv}, we can in \(O(|\Ni_D(P)|)\)~total time store for each article~\(v\in\Ni_D[P]\), the number~\(A[w]\) of articles
  in~\(P\) that are cited by~\(w\). To update the array within the loop in \cref{alg2:forv}, we need \(O(\degi(v))\)~time if \cref{alg2:delv} applies. In total, this is linear time.

  \medskip
  Finally, the pseudocode for solving \hindlE$(\mu)$ is given in \cref{alg:hindlE}.
  For each merged article~\(\mergelt\in\merge\), \cref{alg:hindlE} computes a set~\(\unmerge_P\) and the number~\(\ell_P\) of additional articles~\(v\) with~$\mu({v}) \geq h$ that can be created by extracting.  Then it chooses a set \(\merge^*\) of \(k\)~merged articles~\(\mergelt\in\merge\) with maximum~\(\ell_P\) and, from each \(\mergelt\in\merge^*\), extracts the articles in~\(\unmerge_P\).

  This procedure creates the maximum number of articles that have a $\mu$-value of at least~$h$ while only performing extraction operations on at most $k$ merges.
  
  Obviously, if the solution~$\unmerge$ has at least $h$~parts~$\unmergelt$ with~$\mu(\unmergelt) \geq h$, then we face a yes-instance. Conversely, if the input is a yes-instance, then there are $k$~merged articles that we can apply extraction operations to, such that the resulting partition~$\unmerge$ has at least $h$~parts~$\unmergelt$ with~$\mu(\unmergelt) \geq h$. Since the algorithm produces the maximal number of parts~$\unmergelt$ with~$\mu(\unmergelt) \geq h$, it achieves an \hind{} of at least~$h$.

  The linear running time follows by implementing the check in \cref{lin:check2} in \(O(\degi(v))\)~time as described for \cref{alg:hindkE} and by using counting sort to find the $k$~parts to extract from in \cref{lin:count}.
\end{proof}

\paragraph{Manipulation by dividing.}

Recall that the dividing operation splits a merged article into two arbitrary
parts. First we consider the basic and the conservative case and show that they
are FPT when parameterized by the \hind{}~$h$. Then we show that the
cautious variant is W[1]-hard when parameterized by~$k$.
\hindS{}$(\mu)$ is closely related to 
\hindM$(\mu)$~\citep{BKNSW16,KA13} which is, given a citation graph~$D =
(V, A)$, a subset of articles~$W \subseteq V$, and a non-negative integer~$h$,
to decide whether there is a partition~$\merge$ of~$W$ such that~$\merge$ has
\hind{}~$h$ with respect to~$\mu$. De Keijzer and Apt~\citep{KA13} showed that
\hindM{}$(\scites)$ is NP-hard, even if merges are unconstrained. The NP-hardness of \hindM{}
for~$\mu \in \{\ucites, \mcites\}$ follows. We can reduce \hindM{} to
\hindlS{} by defining the partition~$\merge = \{W\}$, hence we get the following.

\begin{proposition}
\label{prop:hindSNPh}
\hindS{} and \hindlS{} are NP-hard for~$\mu\in\{\scites,\ucites,\mcites\}$.
\end{proposition}

\noindent As to computational tractability, \hindS{} and \hindlS{} are FPT when
parameterized by~$h$---the \hind{} to achieve.

\begin{theorem}
\label{thm:hindS}
\hindS{} and \hindlS{}$(\mu)$ can be solved in $2^{O(h^4\log h)}\cdot
n^{O(1)}$ time, where~$h$ is the \hind{} to achieve
and~$\mu\in\{\scites,\ucites\}$.
\end{theorem}

\begin{algorithm}[t]
  \DontPrintSemicolon
  
  \SetKwFunction{Merge}{Merge}
  
  \KwIn{A citation graph~$D=(V,A)$, a set
  $W\subseteq V$ of articles, a partition~$\merge$ of~$W$,
  nonnegative integers~$h$ and~$k$, and a measure~$\mu$.}
  
  \KwOut{\true{} if $k$ dividing operations can be applied to $\merge$ to yield \hind{}~$h$ and \false{} otherwise.}
  
  \BlankLine
   
  \ForEach{$\mergelt \in \merge$}
  {
    $D'\gets{}$ The graph obtained from~$D$ by removing all citations~$(u, v)$ such that~$v \notin P$ and adding $h+1$~articles~$r_1, \ldots, r_{h+1}$\nllabel{lin:hindm-inst1}\;
  	$W' \gets \mergelt$, $\ell_\mergelt \gets 0$\nllabel{lin:hindm-inst2}\;
  	

  	\For{$i\leftarrow 0$ \KwTo $h$} {
		\If{\Merge{$D', W', h, \mu$}} {
		$\ell_\mergelt \leftarrow h-i$\;
		Break\;
		}
                Add $r_i$ to $W'$ and 
                add each citation~$(r_i, r_j)$, $j \in \{1, \ldots, h + 1\} \setminus \{i\}$ to~$D'$\nllabel{lin:addartart}\;
  	}
  	
  }
  \KwRet{$\exists \merge' \subseteq \merge \text{ s.t. } |\merge'| \leq k 
  \text{ and } \sum_{\mergelt \in \merge'} \ell_\mergelt \geq h$}
  
  \caption{Conservative Dividing}
  \label{alg:hindlS}
\end{algorithm}
\begin{proof}
  The pseudocode is given in \cref{alg:hindlS}. Herein, \Merge{$D,W,h,\mu$} decides \hindM{}$(\mu)$, that is, it returns \true{} if there is a partition~$\mathcal{Q}$ of~$W$ such that~$\mathcal{Q}$ has \hind{}~$h$ and \false{} otherwise. It follows from \citet[Theorem~7]{BKNSW16} that \Merge\ can be carried out in $2^{O(h^4\log h)}\cdot n^{O(1)}$ time.

  \cref{alg:hindlS} first finds, using
  \Merge, the maximum number~$\ell_\mergelt$ of (merged) articles with
  at least~$h$ citations that we can create in each part~$\mergelt \in
  \merge$. For this, we first prepare an instance~$(D', W', h, \mu)$ of \hindM$(\mu)$ in
  \cref{lin:hindm-inst1,lin:hindm-inst2}. In the resulting instance, we ask whether there is a partition of~$P$ with \hind{}~$h$. If this is the case, then we set~$\ell_P$ to~$h$ and, otherwise, we add one artificial article with~$h$ citations to~$W'$ in \cref{lin:addartart}. Then we use \Merge\ again and we iterate this process until \Merge\ returns \true, or we find that there is not even one merged article contained in~$P$ with $h$~citations. Clearly, this process correctly computes~$\ell_P$. Thus, the algorithm is correct. The running time is clearly dominated by the calls to \Merge. Since \Merge\ runs in $2^{O(h^4\log h)}\cdot n^{O(1)}$ time~\cite[Theorem~7]{BKNSW16}, the running time bound follows.
\end{proof}

\noindent 
We note that \Merge\ can be modified so that it outputs the desired partition. Hence, we can modify \cref{alg:hindlS} to output the actual
solution.  Furthermore, for $k=n$, \cref{alg:hindlS} solves the non-conservative
variant, which is therefore also fixed-parameter tractable parameterized
by~\(h\).

 In contrast, for the cautious variant we show W[1]-hardness when parameterized by $k$, the number of allowed operations.

\begin{theorem}
\label{thm:hindkS}
\hindkS{}$(\mu)$ is NP-hard and W[1]-hard when parameterized by~$k$ 
for~$\mu\in\{\scites,\allowbreak \ucites, \mcites\}$, even if the citation graph is
acyclic.
\end{theorem}

\begin{proof}
We reduce from the \textsc{Unary Bin Packing} problem: given a set~$S$ of~$n$
items with integer sizes~$s_i$,~$i \in \{1, \ldots, n\}$, $\ell$~bins and a maximum bin
capacity~$B$, can we distribute all items into the~$\ell{}$ bins? Herein, all sizes are encoded in unary. \textsc{Unary Bin Packing} parameterized by~$\ell$ is W[1]-hard~\citep{jansen2013bin}.

Given an instance~$(S, \ell{}, B)$ of \textsc{Unary Bin Packing}, we produce an instance~$(D, W, \merge, h, \ell{}-1)$
of \hindkS$(\scites)$. Let~$s^* = \sum_i
s_i$ be the sum of all item sizes. We assume that~$B < s^*$ and~$\ell{}\cdot B \geq
s^*$ as, otherwise, the problem is trivial, since all items fit into one bin or
they collectively cannot fit into all bins, respectively. Furthermore, we assume that $\ell{} < B$
since, otherwise, the instance size is upper bounded by a function of $\ell{}$ and, hence, is trivially FPT with respect to~$\ell{}$. We construct the instance of \hindkS$(\scites)$ in polynomial time as follows.
\begin{itemize}
  \item Add~$s^*$ articles~$x_1,\ldots, x_{s^*}$ to~$D$. These are only used to
  increase the citation count of other articles.
  \smallskip
  \item Add one article~$a_i$ to~$D$ and~$W$ for each~$s_i$.
  \item For each article~$a_i$, add citations~$(x_j, a_i)$ for all~$1\leq j \leq
  s_i$ to~$G$. Note that, after adding these citations, each article~$a_i$ has citation count~$s_i$.
  \smallskip
  \item Add~$\Delta \coloneqq \ell{}\cdot B - s^*$ articles~$u_1, \ldots, u_\Delta$ to~$D$
  and~$W$.
  \item For each article~$u_i$ with $i\in\{1, \ldots, \Delta\}$, add an citation~$(x_1, u_i)$ to~$D$.  Note that each
  article~$u_i$ has citation count~$1$.
  \smallskip
  \item Add~$B-\ell{}$ articles~$h_1, \ldots, h_{B-\ell{}}$ to~$D$
  and~$W$.
  \item For each article~$h_i$ with $i\in\{1, \ldots, B-\ell\}$, add citations~$(x_j, h_i)$ for all~$1\leq j \leq B$
  to~$D$. Note that each article~$h_i$ has citation count~$B$.
  \smallskip
  \item Add~$\mergelt^* = \{a_1, \ldots, a_n, u_1, \ldots,
  u_\Delta\}$ to~$\merge$, for each article~$h_i$ with $i\in\{1, \ldots, B-\ell\}$, add~$\{h_i\}$ to~$\merge$, and set~$h = B$.
\end{itemize}
Now we show that~$(S, \ell{}, B)$ is a yes-instance if and only if~$(D, W, \merge, h,
\ell{}-1)$ is a yes-instance.

($\Rightarrow$) Assume that~$(S, \ell{}, B)$ is a yes-instance and let~$S_1, \ldots,
S_\ell{}$ be a partition of~$S$ such that items in~$S_i$ are placed in bin~$i$. Now
we split~$\mergelt^*$ into~$\ell{}$ parts~$\unmergelt_1, \ldots, \unmergelt_\ell{}$ in the
following way. Note that for each~$S_i$, we have that~$\sum_{s_j \in S_i} s_j =
B - \delta_i$ for some~$\delta_i \geq 0$. Furthermore, $\sum_i \delta_i = \Delta$. Recall
that there are~$\Delta$ articles~$u_1, \ldots, u_\Delta$ in~$\mergelt^*$.
Let~$\delta_{<i}=\sum_{j<i}\delta_j$ and~$U_i = \{
u_{\delta_{<i}+1},\ldots, u_{\delta_{<i}+\delta_i}\}$, with~$\delta_0 = 0$ and
if~$\delta_i > 0$, let~$U_i = \emptyset$ for $\delta_i = 0$. We set
$\unmergelt_i = \{a_j \ | \ s_j \in S_i \} \cup U_i$. Then for
each~$\unmergelt_i$, we have that \(\scites(\unmergelt_i) = \scites(\{a_j \ | \
s_j \in S_i \}) + \scites(U_i),\) which simplifies to
\(\scites(\unmergelt_i) = \sum_{s_j \in S_i} s_j + \delta_i = B.\)
For each $i$, $1 \leq i \leq n$, we have $\scites(\{h_i\}) = B$. Hence, $\unmerge = \{\unmergelt_1, \ldots,
\unmergelt_\ell{}, \{h_1\}, \ldots, \{h_{B-\ell{}}\}\}$ has \hind{}~$B$.

($\Leftarrow$) Assume that~$(D, W, \merge, h, \ell{}-1)$ is a yes-instance and
let~$\unmerge$ be a partition with \hind{}~$h$. Recall that~$\merge$ consists of~$\mergelt^*$
and~$B-\ell{}$ singletons~$\{h_1\}, \ldots, \{h_{B-\ell{}}\}$, which are hence also contained in~$\unmerge$. Furthermore,~$\scites(\{h_i\}) = B$ for each~$h_i$ and, by the definition of the
\hind{}, there are~$\ell{}$ parts~$\unmergelt_1, \ldots, \unmergelt_\ell{}$
with~$\unmergelt_i \subset \mergelt^*$ and~$\scites(\unmergelt_i) \geq B$
for each~$i$. Since, by definition, $\scites(\mergelt^*) = \ell{}\cdot B$
and~$\scites(\mergelt^*) = \sum_{1\leq i\leq \ell{}} \scites(\unmerge_i)$ we have
that~$\scites(\unmergelt_i) = B$ for all~$i$. It follows
that~$\scites(\unmergelt_i \setminus \{u_1, \ldots, u_\Delta\}) \leq B$ for all~$i$. This implies that packing into bin~$i$ each item in $\{s_j \ | \ a_j \in \unmergelt_i\}$ solves the 
instance~$(S, \ell{}, B)$.

Note that this proof can be modified to cover also the~$\ucites$ and
the~$\mcites$ case by adding~$\ell{}\cdot s^*$ extra~$x$-articles and ensuring that no two articles in $W$ are cited by the same $x$-article.
\end{proof}

\section{fusionCite}
We now consider the \(\mcites\) measure, which makes manipulation
considerably harder than the other measures. In particular,
we obtain that, even in the most basic case, the manipulation
problem is NP-hard.

\begin{theorem}\label{thm:aehard}\label{thm:hindANP}
    \hindA{}$(\mcites)$ and \hindE{}$(\mcites)$ are NP-hard, even if the citation graph is acyclic and \(s=3\), where \(s\)~is the largest number of articles merged into one.
\end{theorem}
\begin{figure}
  \centering
  \includegraphics{split-index-1}
  \caption{Illustration of the construction in the proof of \cref{thm:aehard} for a literal~\(\neg x_i\) contained in a clause~\(c_j\).}
  \label{fig:aehard}
\end{figure}
\begin{proof}
  We reduce from the NP-hard \textsc{3-Sat} problem: given a 3-CNF formula~$F$ with $n$~variables and $m$~clauses, decide whether $F$~has a satisfying truth assignment to its variables.  Without loss of generality, we assume $n+m>3$ and that each clause contains three literals over mutually distinct variables.  Given a formula~$F$ with variables~$x_1,\dots,x_n$ and clauses~$c_1,\dots,c_m$ such that $n+m>3$, we produce an instance~$(D,W,\merge,m+n)$ of \hindA{}$(\mcites)$ or \hindE{}$(\mcites)$ in polynomial time as follows.  The construction is illustrated in \cref{fig:aehard}.

  For each variable~$x_i$ of~$F$, add to~$D$ and~$W$ sets~$\mathcal X_i^F:=\{X_{i,1}^{F},X_{i,2}^F,X_{i,3}^F\}$ and~\(\mathcal X_i^T:=\{X_{i,1}^T,X_{i,2}^T,X_{i,3}^T\}\) of \emph{variable articles}.  Add \(\mathcal X_i^F\) and~\(\mathcal X_i^T\) to~\(\merge\).  Let \(h:=m+n\).  For each variable~\(x_i\), add
  \begin{enumerate}
  \item \(h-2\) citations from (newly-introduced) distinct atomic articles to
    \(X_{i,1}^T\) and \(X_{i,1}^F\),
  \item citations from \(X_{i,1}^F\) to \(X_{i,2}^T\) and from \(X_{i,2}^T\) to \(X_{i,3}^F\), and
  \item citations from \(X_{i,1}^T\) to \(X_{i,2}^F\) and from \(X_{i,2}^F\) to \(X_{i,3}^T\).
\end{enumerate}
Next, for each clause~$c_j$ of~$F$, add a \emph{clause article}~$C_j$ with \(h-4\) incoming citations to~$D$, to~$W$, and add~$\{C_j\}$ to~$\merge$.  Finally, if a positive literal~$x_i$ occurs in a clause~$c_j$, then add citations $(X_{i,\ell}^T,C_j)$ to~$D$ for~$\ell\in\{2,3\}$.  If a negative literal~$\neg x_i$ occurs in a clause~$c_j$, then add citations $(X_{i,\ell}^F,C_j)$ to~$D$ for~$\ell\in\{2,3\}$.  This concludes the construction.  Observe that $D$~is acyclic since all citations go from variable articles to clause articles or to variable articles with a higher index.  It remains to show that $F$~is satisfiable if and only if $(D,W,\merge,h)$~is a yes-instance.

  ($\Rightarrow$) If $F$~is satisfiable, then a solution~$\unmerge$ for $(D,W,\merge,h)$~looks as follows: for each $i \in \{1, \ldots, n\}$, if $x_i$~is true, then we put $X_i^F\in\unmerge$ and we put $\mathcal X_i^T\in\unmerge$ otherwise.  All other articles of~$D$ are added to~$\unmerge$ as singletons.  We count the citations that every part of~$\unmerge$ gets from other parts of~$\unmerge$.  If $x_i$~is true, then $\mathcal X_i^F$~gets two citations from $\{X_{i,\ell}^T\}$ for~$\ell\in\{1,2\}$ and the \(h-2\) initially added citations.  Moreover, for the clause~$c_j$ containing the literal~$x_i$, $\{C_j\}$~gets two citations from $\{X_{i,\ell}^T\}$ for~$\ell\in\{2,3\}$, at least two citations from variable articles for two other literals it contains, and the \(h-4\) initially added citations.  Symmetrically, if $x_i$~is false, then $\{\mathcal X_i^T\}$ gets $h$~citations and so does every~$\{C_j\}$ for each clause~$c_j$ containing the literal~$\neg x_i$.  Since every clause is satisfied and every variable is either true or false, it follows that each of the $m$~clause articles gets $h$~citations and that, for each of the $n$~variables $x_i$, either $\mathcal X_i^F$ or~$\mathcal X_i^T$ gets $h$~citations.  It follows that $h=m+n$~parts of~$\unmerge$ get at least $h$~citations and thus, that $\unmerge$ has \hind{} at least~$h$.

  (\(\Leftarrow\)) Let $\unmerge$~be a solution for $(D,W,\merge,m+n)$.  We first show that, for each variable~$x_i$, we have either $\mathcal X_i^T\in\unmerge$ or $\mathcal X_i^F\in\unmerge$.  To this end, it is important to note two facts:
  \begin{enumerate}
  \item For each variable~$x_i$, $\mathcal X_i^T$ contains two atomic articles with one incoming arc in~\(D\) and one with \(h-2\)~incoming arcs.  Thus, no subset of \(\mathcal X_i^T\) can get $h$~citations.  The same holds for \(\mathcal X_i^F\).
  \item If, for some variable~$x_i$, the part~$\mathcal X_i^T\in\unmerge$ gets $h$~citations, then ${\cal X}_i^F\notin\unmerge$ and vice versa.
  \end{enumerate}
  Thus, since there are at most $m$~clause articles and $\unmerge$~contains $h=m+n$~parts with $h$~citations, $\unmerge$~contains exactly one of the parts~$\mathcal X_i^T,\mathcal X_i^F$ of each variable~$x_i$.  It follows that, in $\unmerge$, all singleton clause articles have to receive $h$~citations.  Each such article gets at most \(h-4\)~initially added citations and citations from at most three sets \(\mathcal X_i^T\) or \(\mathcal X_i^F\) for some variable~\(x_i\).  Thus, for each clause~$c_j$, there is a literal~$x_i$ in~$c_j$ or a literal~$\neg x_i$ in~$c_j$ such that $\mathcal X_i^T\notin\unmerge$ or $\mathcal X_i^F\notin\unmerge$, respectively.  It follows that setting each $x_i$~to true if and only if $\mathcal X_i^T\notin\unmerge$ gives a satisfying truth assignment to the variables of~$F$.
\end{proof}

\noindent This NP-hardness result motivates the search for fixed-parameter tractability.
\begin{theorem}
\label{thm:hindAmcites}
\hindA{}$(\mcites)$ can be solved in $O(4^{h^2}(n + m))$ time, where
$h$~is the \hind{} to achieve.
\end{theorem}

\begin{proof}
We use the following procedure to solve an instance~$(D,W,\merge,h)$ of \hindA{}$(\mcites)$. 

Let $\merge_{\geq h}$~be the set of merged articles~$\mergelt\in\merge$ with $\mcites(\mergelt)\geq h$.  If $|\merge_{\geq h}|\geq h$, then we face a yes-instance and output ``yes''. To see that we can do this in linear time, note that, given~$\merge$, we can compute~$\mcites(P)$ in linear time for each~$P \in \merge$. Below we assume that $|\merge_{\geq h}|<h$.

 First, we atomize all $\mergelt\in\merge$ that cannot have $h$~or more citations, that is, for which, even if we atomize all merged articles except for $\mergelt$, we have $\mcites(\mergelt) < h$.  Formally, we atomize~\(\mergelt\) if
\(
\sum_{v\in\mergelt}|\Ni_{D-\mergelt}(v)|<h.
\)
Let $\merge'$ be the partition obtained from $\merge$ after these atomizing
operations; note that $\merge'$~can be computed in linear time.

 The basic idea is now to look at all remaining merged articles that receive at least~$h$ citations from atomic articles; they form the set~$\merge_{<h}$ below. They are cited by at most~$h - 1$ other merged articles. Hence, if the size of $\merge_{<h}$ exceeds some function~$f(h)$, then, among the contained merged articles, we find a large number of merged articles that do not cite each other. If we have such a set, then we can atomize all other articles, obtaining \hind{}~$h$. If the size of $\merge_{<h}$ is smaller than $f(h)$, then we can determine by brute force whether there is a solution. 

Consider all merged articles $\mergelt\in\merge'$ that have less than~$h$
citations but can obtain $h$ or more citations by applying atomizing operations
to merged articles in $\merge'$. Let us call the set of
these merged articles $\merge_{<h}$. Formally, $\mergelt\in\merge_{<h}$ if
$\sum_{v\in\mergelt}|\Ni_{D-\mergelt}(v)| \geq h$ and $\mcites(P) < h$.
Again, $\merge_{<h}$ can be computed in linear time. Note that
$\merge'\setminus(\merge_{\geq h}\cup\merge_{<h})$ consists only of
singletons.

Now, we observe the following. If there is a set~$\merge^*\subseteq\merge_{<h}$ of at least \(h\)~merged articles such that, for all $\mergelt_i, \mergelt_j\in\merge^*$, neither $\mergelt_i$ cites $\mergelt_j$ nor $\mergelt_j$ cites $\mergelt_i$, then we can atomize all merged articles in~$\merge'\setminus\merge^*$ to reach an \hind{} of at least~$h$.  We finish the proof by showing that we can conclude the existence of the set~\(\merge^*\) if \(\merge_{<h}\) is sufficiently large and solve the problem using brute force otherwise.

Consider the undirected graph $G$ that has a vertex $v_\mergelt$ for each~$\mergelt\in\merge_{<h}$ and an edge between $v_{\mergelt_i}$ and $v_{\mergelt_j}$ if $\mergelt_i$ cites~$\mergelt_j$ or $\mergelt_j$ cites~$\mergelt_i$. Note that $\{v_\mergelt\mid\mergelt\in\merge^*\}$ forms an independent set in $G$. Furthermore, let $I$ be an independent set in $G$ that has size at least $h$. Let $\merge^{**} = \{\mergelt\in\merge_{<h}\mid v_\mergelt \in I\}$. Then, we can atomize all merged articles in $\merge'\setminus\merge^{**}$ to reach an \hind{} of at least $h$.
 
We claim that the number of edges in $G$ is at most $(h-1)\cdot |\merge_{<h}|$. This is because the edge set of~$G$ can be enumerated by enumerating for every vertex~$v_\mergelt$ the edges incident with~$v_\mergelt$ that result from a citation of~$\mergelt$ from another $\mergelt'\in\merge_{<h}$. The citations for each~$\mergelt$ are less than~$h$ as, otherwise, we would have that $\mergelt\in\merge_{\geq h}$. Now, we can make use of Tur\'{a}n's Theorem, which can be stated as follows: If a graph with $\ell$ vertices has at most $\ell k/2$~edges, then it admits an independent set of size at least~$\ell/(k+1)$~\cite[Exercise 4.8]{Juk01}. 
Hence, if $|\merge_{<h}| \geq 2h^2-h$, then we
face a yes-instance and we can find a solution by taking an arbitrary
subset~$\merge'_{<h}$ of $\merge_{<h}$ with $|\merge'_{<h}| = 2h^2-h$, by atomizing
every merged article outside of~$\merge'_{<h}$, and by guessing which merged articles we need to atomize inside of~$\merge'_{<h}$.
If $|\merge_{<h}| < 2h^2-h$, then we guess which merged articles in
$\merge_{<h}\cup\merge_{\geq h}$ we need to atomize to obtain a solution if it
exists. In both cases, for each guess we need linear time to determine
whether we have found a solution, giving the overall running time of
$O(4^{h^2} \cdot (m + n))$.
\end{proof}

\noindent For the conservative variant, however, we cannot achieve FPT, even if we add the number of atomization operations and the maximum size of a merged article to the parameter.
\begin{theorem}\label{thm:hindkAW1}
  \hindlA{}$(\mcites)$ is NP-hard and W[1]-hard when parameterized by $h+k+s$, where
  $s:=\max_{\mergelt \in \merge}|\mergelt|$, even if the citation graph is acyclic.
\end{theorem}
\begin{proof}
  We reduce from the \textsc{Clique} problem: given a graph~$G$ and an integer~$k$, decide whether $G$~contains a clique on at least $k$~vertices.  \textsc{Clique} parameterized by~$k$ is known to be W[1]-hard.  

Given an instance~$(G,k)$ of \textsc{Clique}, we produce an instance~$(D,W,\merge,h,k)$ of \hindlA{}$(\mcites)$ in polynomial time as follows.  Without loss of generality, we assume~$k\geq 4$ so that ${k \choose 2}\geq 4$. For each vertex~$v$ of~$G$, introduce a set~$R_v$ of $\lceil {k\choose 2}/2\rceil$~vertices to~$D$ and~$W$ and add~$R_v$ as a part to~$\merge$.  For an edge~$\{v,w\}$ of~$G$, add to~$D$ and~$W$ a vertex~$e_{\{v,w\}}$ and add $\{e_{\{v,w\}}\}$ to~$\merge$.  Moreover, add a citation from each vertex in~$R_v\cup R_w$ to $e_{\{v,w\}}$.  Finally, set $h:={k\choose 2}$.  Each of $h$, $k$ and~$s$ in our constructed instance of \hindlA{}$(\mcites)$ depends only on~$k$ in the input \textsc{Clique} instance.  It remains to show that $(G,k)$~is a yes-instance for \textsc{Clique} if and only if $(D,W,\merge,h,k)$~is.

($\Rightarrow$) Assume that $(G,k)$~is a yes-instance and let $S$~be a clique in~$G$.  Then, atomizing~$R_v$ for each~$v\in S$ yields ${k\choose 2}$~articles with at least ${k\choose 2}$~citations in~$D$: for each of the ${k\choose 2}$~pairs of vertices~$v,w\in S$, the vertex~$e_{\{v,w\}}$ gets $\lceil {k\choose 2}/2\rceil$~citations from the vertices in~$R_v$ and the same number of citations from the vertices in~$R_w$ and, thus, at least ${k\choose 2}$~citations in total.

($\Leftarrow$) Assume that $(D,W,\merge,h,k)$~is a yes-instance and let $\unmerge$~be a solution.  We construct a subgraph~$S=(V_S,E_S)$ of~$G$ that is a clique of size~$k$.  Let $V_S:=\{v\in V(G)\mid R_v\in\merge\setminus\unmerge\}$ and $E_S:=\{\{v,w\}\in E(G)\mid \{v,w\}\subseteq V_S\}$, that is, $S=G[V_S]$. Obviously,  $|V_S|\leq k$. It remains to show $|E_S|\geq{k\choose 2}$, which implies both that $|V_S|=k$ and that $S$~is a clique.  To this end, observe that the only vertices with incoming citations in~$D$ are the vertices~$e_{\{v,w\}}$ for the edges~$\{v,w\}$ of~$G$.  The only citations of a vertex~$e_{\{v,w\}}$ are from the parts~$R_v$ and~$R_w$ in~$\merge$.  That is, with respect to the partition~$\merge$, each vertex~$e_{\{v,w\}}$ has two citations.  Since the \hind{}~$h$ to reach is ${k\choose 2}$, at least ${k\choose 2}$~vertices $e_{\{v,w\}}$ have to receive ${k\choose 2}\geq 4$~citations, which is only possible by atomizing both~$R_v$ and~$R_w$.  That is, for at least ${k\choose 2}$~vertices $e_{\{v,w\}}$, we have $\{R_v,R_w\}\subseteq\merge\setminus\unmerge$ and, thus, $v,w\subseteq V_S$ and $\{v,w\}\in E_S$.  It follows that $|E_S|\geq{k\choose 2}$.
\end{proof}
\noindent The reduction given above easily yields the same hardness result for most other problem variants: a vertex~$e_{\{v,w\}}$ receives a sufficient number of citations only if~$R_v$ and~$R_w$ are atomized. Hence, even if we allow  extractions or divisions on~$R_v$, it helps only if we extract or split off all articles in~$R_v$. The only difference is that the number of allowed operations is set to~$k\cdot (\lceil {k\choose 2}/2 -1)\rceil$ for these two problem variants. By the same argument, we obtain hardness for the conservative variants.
\begin{corollary}
\label{cor:mcites}
 For $\mu = \mcites$, 
  \hindlE{}$(\mu)$, \hindkE{}$(\mu)$,
  \hindlS{}$(\mu)$, and~\hindkS{}$(\mu)$ are NP-hard and W[1]-hard when
  parameterized by $h+k+s$, where $s:=\max_{\mergelt \in \merge}|\mergelt|$, even if the citation graph is acyclic.
\end{corollary}

\section{Computational experiments}
To assess how much the \hind{} of a researcher can be manipulated by
splitting articles in practice, we performed computational experiments with data
extracted from Google Scholar.

\paragraph{Description of the data.} We use three data sets collected 
by
\citet{BKNSW16}. 
One data set consists of 22~selected authors of the conference IJCAI'13. The
selection of these authors was biased to obtain profiles of authors in
their early career. More precisely, the selected authors have a Google
Scholar profile, an \hind{} between 8 and~20, between 100 and 1000
citations, and have been active between 5 and 10 years when the data was collected. Below we refer to this dataset by \dsijcai. The other two
data sets contain Google Scholar data of `AI's 10 to Watch', a list of
young accomplished researchers in AI 
compiled 
by \textit{IEEE Intelligent Systems}. One data set contains five
profiles from the 2011 edition (\dsaia), the other eight profiles from the 2013
edition of the list (\dsaib). In comparison to \citet{BKNSW16} we removed one author from the \dsaib\ data set because the data were inconsistent. All data were gathered between November 2014 and
January 2015. For an overview on the data, see
\cref{tab:data-sets}.
\begin{table}[t]\centering
  \caption{Properties of the three data sets. Here, $p$~is the number of profiles for each data set, $\overline{|W|}$~is the average number of atomic articles, $\overline{c}$~is the average number of citations, $\overline{h}$~is the average \hind{} in the data set (without merges), and \(h/a\)~is the average \hind{} increase per year; the `$\mathrm{max}$' subscript denotes the maximum of these values.}
  \begin{tabular}[t]{@{}rrrrrrrrrr@{}}
\toprule
 & \multicolumn{1}{c}{$p$} & \multicolumn{1}{c}{$\overline{|W|}$}  & \multicolumn{1}{c}{$|W|_{\mathrm{max}}$} & \multicolumn{1}{c}{$\overline{c}$}  & \multicolumn{1}{c}{$c_{\mathrm{max}}$} & \multicolumn{1}{c}{$\overline{h}$}  & \multicolumn{1}{c}{$h_{\mathrm{max}}$} &\(h/a\) \\ \midrule

\dsaia    & 5 & 170.2 & 234 & 1614.2 & 3725 & 34.8 & 46 & 2.53\\
\dsaib    & 7 & 58.7 & 144 & 557.5 & 1646 & 14.7 & 26 & 1.57 \\
\dsijcai  & 22 & 45.9 & 98 & 251.5 & 547 & 10.4 & 16 & 1.24 \\ \bottomrule
\end{tabular}
\label{tab:data-sets}
\end{table}

Due to difficulties in obtaining the data from Google Scholar,
\citet{BKNSW16} did not gather the concrete set of citations for
articles which are cited a large number of times. These were articles
that will always be part of the articles counted in the \hind. They
subsequently ignored these articles since it is never beneficial to
merge them with other articles in order to increase the \hind. In our
case, although such articles may be merged initially, they will also
always be counted in the \hind\ and hence their concrete set of
citations is not relevant for us as well. The information about
whether such articles could be merged is indeed contained in the
datasets.

\paragraph{Generation of profiles with merged articles.} In our
setting, the input consists of a profile which already contains some
merged articles. The merges should be performed in a way which
reflects the purpose of merging in the Google Scholar interface. That
is, the merged articles should roughly correspond to different
versions of the same work. To find different versions of the same
work, we used the compatibility graphs for each profile provided by
\citet{BKNSW16} which they generated as follows. The set of vertices
is the set of articles in the profile. For each article~$u$ let $T(u)$
denote the set of words in its title. There is an edge between
articles~$u$ and~$v$ if~$|T(u)\cap T(v)|\ge t\cdot |T(u)\cup T(v)|$,
where~$t\in [0,1]$ is the \emph{compatibility threshold}. For $t=0$,
the compatibility graph is a clique; for~$t=1$ only articles with the
same words in the title are adjacent. For~$t\le 0.3$, very dissimilar
articles are still considered compatible~\citep{BKNSW16}. Hence, we
usually focus on $t \ge 0.4$ below.

We then generated the merged articles as follows. We used four different
methods so that we can avoid artifacts that could be introduced by one
specific method. Each method iteratively computes an inclusion-wise
maximal clique~$C$ in the compatibility graph~$D$, adds it as a merged
article to the profile, and then removes~$C$ from~$D$. The clique~$C$
herein is computed as follows.
\begin{enumerate}
\item[\mmgreedymax] Recursively include into~$C$ a largest-degree vertex
  which is adjacent to all vertices already included until no such
  vertex exists anymore.
\item[\mmgreedymin] Recursively include into~$C$ a smallest-degree
  vertex which is adjacent to all vertices already included until no
  such vertex exists anymore.
\item[\mmmax] A maximum-size clique.
\item[\mmramsey] A recursive search of a maximal clique in the
  neighborhood of a vertex~$v$ and the remaining graph. See algorithm
  \textsf{Clique Removal} by \citet{boppana_approximating_1992} for
  details.
\end{enumerate}
If the compatibility graph has no edge anymore, then each method adds
all remaining articles as atomic articles of the profile.

\begin{figure}[!p]
  \centering
  \include{initial_hindex_distributions}
  \caption{\Hind\ distributions of the profiles with generated merged
    articles in comparison to the profiles without any merged articles.}
  \label{fig:initial-hindices}
\end{figure}

\cref{fig:initial-hindices} shows the distributions of the \hinds\ of
the generated profiles with merged articles and those where no article
has been merged. The lower edge of a box is the first quartile, the
upper edge the third quartile, and the thick bar is the median; the
remaining data points are shown by dots. Note that when no article is
merged---and no atomic article cites itself---all three citation
measures coincide. Often, merging compatible articles leads to a
decline in \hind\ in our datasets and this effect is most pronounced
for the more senior authors (in \dsaia). In contrast, merging very
closely related articles (compatibility threshold $t = 0.9$) for
authors in \dsaib\ led to increased \hinds. The initial \hinds\ are
very weakly affected by the different methods for generating initially
merged articles.

\paragraph{Implementation.} We implemented
\cref{alg:hindkA,alg:hindkE,alg:hindlE}---the exact, linear-time algorithms from \cref{sec:sucites} for \hindlA{}, \hindlE{}, and \hindkE{}, respectively, each for all three citation measures, $\scites$, $\ucites$, and $\mcites$.
The algorithms for $\scites$ and $\ucites$ were implemented directly as described.
For $\mcites$, we implemented minor modifications of the described algorithms---to make the algorithms well-defined, whenever we apply $\mcites$, we need to specify which article partition is used, that is, which articles are currently merged, but otherwise the basic algorithms are unchanged.
More specifically, recall that for \cref{alg:hindkA} we greedily perform atomizing operations if they increase the \hind.
Thus, in the adaption to $\mcites$, the partition $\mergelt$ is continuously updated whenever the check in \cref{lin:recfin} is positive, and the application of $\mu = \mcites$ in that line uses the partition~$\mergelt$ which is current at the time of application.
Similarly, the partitions are updated after positive checks in \cref{alg:hindkE}, \cref{lin:check1}, and in \cref{alg:hindlE}, \cref{lin:check2}, and the application of $\mu = \mcites$ in that line uses the current partition $\mergelt$.

Using the algorithms, we computed \hind{} increases under the respective restrictions.
For $\scites$ and $\ucites$ these algorithms yield the maximum-possible \hind{} increases by \cref{thm:hindA,thm:hindE}.
For $\mcites$, we only obtain a lower bound.

The implementation is in Python 3.6.7 under Ubuntu Linux 18.04 and the
source code is freely available.\footnote{See \url{http://fpt.akt.tu-berlin.de/hindexexp}.} In total,
137,626 instances of the decision problems were generated. Using an
Intel Core i5-7200U CPU with 2.5\,GHz and 8\,GB RAM, the instances
could be solved within fourteen hours altogether (ca.\ 350\,ms average
time per instance).

\begin{figure}[!p]
  \centering
  \include{change_by_merging}
  \caption{Number of profiles whose \hinds\ may be increased by unmerging.}
  \label{fig:potential}
\end{figure}

\paragraph{Authors with Potential for Manipulation.}\cref{fig:potential} gives the number of profiles in which the
\hind\ can be increased by unmerging articles. We say the profiles or
the corresponding authors have \emph{potential}. Concordant with
intuition, for each threshold value, the methods for creating initial
merges are roughly ordered according to the number of authors with
potential as follows:
$\mmmax > \mmramsey > \mmgreedymax > \mmgreedymin$. \mmgreedymax\ and
\mmgreedymin\ are surprisingly close. However, the differences between
the methods in general are rather small, indicating that the property
of having potential is inherent to the profile rather than the method
for generating initial merges. Since \mmgreedymax\ is one of the most
straightforward of the four, 
we will focus only on
\mmgreedymax\ below.

At first glance, we could expect that the number of authors with
potential decrease monotonously with increasing compatibility
threshold: Note that, for increasing compatibility threshold the edge
sets in compatibility graphs are decreasing in the subset order. Hence
each maximal clique in the compatibility graph can only increase in
size. However, since we employ heuristics to find the set of initial
merges (in the case of \mmramsey, \mmgreedymax, and \mmgreedymin) and
since there may be multiple choices for a maximum-size clique (for
\mmmax), different possible partitionings into initial merges may
result. This can lead to the fact that the authors with potential do
not decrease monotonously with increasing compatibility threshold.

Furthermore, with the same initial merges it can happen that an
increase in the \hind\ value through unmerging with respect to $\scites$ is
possible and no increase is possible with respect to $\ucites$ and
vice versa. The first may happen, for example, if two articles $v, w$
are merged such that $\scites(\{v, w\})$ is above but
$\ucites(\{v, w\})$ is below the \hind\ threshold. The second may
happen if the \hind\ of the merged profile is lower for $\ucites$
compared to that for $\scites$. Then, unmerging articles may yield
atomic articles which are still above the \hind\ threshold for
$\ucites$ but not for~$\scites$. As can be seen from
\cref{fig:potential}, both options occur in our data set.

The fraction of authors with potential differs clearly between the
three data sets. The authors in \dsaia\ have already accumulated so
many citations that almost all have potential for each threshold up
to~$0.6$. Meanwhile, the authors with potential in \dsaib\ continually
drop for increasing threshold and this drop is even more pronounced
for \dsijcai. This may reflect the three levels of seniority
represented by the datasets.

There is no clear difference between the achievable \hinds\ when
comparing $\mcites$ to $\ucites$ and $\scites$: While there are
generally more authors with potential for each threshold for $\mcites$
in the \dsaia\ dataset, there are less authors with potential for the
\dsaib\ dataset, and a similar number of authors with potential for
the \dsijcai\ dataset.

Focusing on the most relevant threshold,~$0.4$, and the
$\ucites$-measure, which is used by Google Scholar~\citep{BKNSW16}, we
see that all authors (100\,\%) in \dsaia\ could potentially increase
their \hinds\ by unmerging, four authors (57\,\%) in \dsaib, and seven
(31\,\%) in \dsijcai. We next focus only on these authors with
potential and gauge to which extent manipulation is possible.

\begin{figure}[!t]
  \centering
  \include{vs-comp}
  \caption{\Hind\ increases for each compatibility
    threshold for authors with potential (note that these authors may
    be different for different threshold values). The increases are
    largest-possible for $\scites$ and $\ucites$.}
  \label{fig:threshold}
\end{figure}

\begin{figure}[!p]
  \centering
  \include{vs-k} 
  \caption{\Hind\ increases versus number of changed articles or
    allowed operations for authors with potential and compatibility
    threshold~$0.4$. The increases are largest-possible for $\scites$
    and $\ucites$.}
  \label{fig:operations}
\end{figure}

\paragraph{Extent and Cost of Possible Manipulation.}
\cref{fig:threshold} shows the largest achievable \hind\ increases for
the authors with potential in the three data sets: again, the lower
edge of a box is the first quartile, the upper edge the third
quartile, and the thick bar is the median; the remaining data points
are shown by dots.

In the majority of cases, drastic increases can only be achieved when the compatibility threshold is lower than~$0.4$.
Generally, the increases achieved for the $\mcites$ measure are slightly lower than for the other two but the median is at most by one smaller.
Because of the heuristic nature of our algorithms for $\mcites$, we cannot exclude the possibility that the largest possible increases for $\mcites$ are comparable to the other two measures.
In the most relevant regime of $\ucites$ and compatibility threshold~$t = 0.4$, the median \hind\ increases are 4 for the \dsaia\ authors, 1 for the \dsaib\ authors, and 2 for the \dsijcai\ authors.
Notably, there is an outlier in \dsijcai\ who can achieve an increase of~5.

\cref{fig:operations} shows the \hind\ increases that can be achieved
by changing a certain number of articles (in the rows containing the
conservative problem variants) or with a certain number of operations
(in the row containing the cautious problem variant) for compatibility
threshold~$0.4$. For the majority of \dsaib\ and \dsijcai\ authors we
can see that, if manipulation is possible, then the maximum \hind\
increase can be reached already by manipulating at most two articles
and performing at most two unmerges. The more senior authors in the
\dsaia\ dataset can still gain increased \hinds\ by manipulating four
articles and performing four unmerges. For the outlier in \dsijcai\
with \hind\ increase of~5, we see that there is one merged article
which contains many atomic articles with citations above her or his
unmanipulated \hind: With respect to an increasing number of operations,
we see a continuously increasing \hind\ for \hindkE\ compared to a
constant high increase for \hindlA.

Summarizing, our findings indicate that realistic profiles from
academically young authors can in the majority of cases not be
manipulated by unmerging articles. If they can, then in the majority
of cases the achievable increase in \hind\ is at most two.
Furthermore, our findings indicate that the increase can be obtained
by tampering with a small number of merged articles (at most two in
the majority of cases).

\section{Conclusion}
\label{sec:concl}
In summary, our theoretical results suggest that using $\mcites$ as a citation measure for merged articles makes manipulation by undoing merges harder.
From a practical point of view, our experimental results 
indicate that author profiles with surprisingly large \hind{} may be 
worth inspecting concerning potential manipulation.

 Regarding theory, we leave three main 
open questions concerning the computational complexity of 
\hindE{}$(\mcites)$, the parameterized complexity of \hindS{}$(\mcites)$, as
well as the parameterized complexity of \hindkS$(\scites / \ucites)$ with respect to $h$
(see \cref{tab:results}), as the most immediate challenges for future work. Also, finding hardness reductions that produce more realistic instances would be desirable. From the experimental side, evaluating the potentially possible \hind{} increase by splitting on real merged profiles would be interesting as well as experiments using $\mcites$ as a measure. 
Moreover, it makes sense to consider the manipulation of the 
\hind{} also in context with the simultaneous manipulation of other 
indices (e.g., Google's i10-index, see also \citet{EP16}) 
and to look for Pareto-optimal solutions. We suspect that our algorithms easily adapt to other indices.
In addition, it is natural to consider combining merging and splitting 
 in manipulation of author profiles.

\section*{Acknowledgements} The authors thank Clemens Hoffmann and
Kolja Stahl for their excellent support in implementing our algorithms
and performing and analyzing experiments.

Manuel Sorge's work was carried out while affiliated with Technische
Universität Berlin, Berlin, Germany, Ben-Gurion University of the
Negev, Beer-Sheva, Israel, and University of Warsaw, Warsaw, Poland.

Hendrik Molter was supported by the DFG, projects DAPA (NI 369/12) and MATE (NI 369/17).
Manuel Sorge was supported by the DFG, project DAPA (NI 369/12), by the People Programme (Marie Curie Actions) of the European Union's Seventh Framework Programme (FP7/2007-2013) under REA grant agreement number {631163.11}, the Israel Science Foundation (grant no.
551145/14), and the European Research Council (ERC) under the European Union’s Horizon 2020 research and innovation programme under grant agreement number 714704. 
Toby Walsh was supported by the European Research Council (ERC) under the European Union’s Horizon 2020 research and innovation programme under grant agreement number 670077.\\ \includegraphics[width=50px]{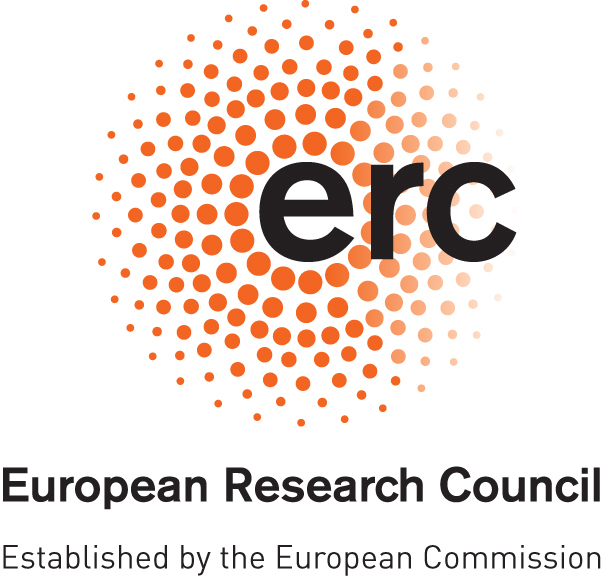}\hspace{.5cm}\includegraphics[width=50px]{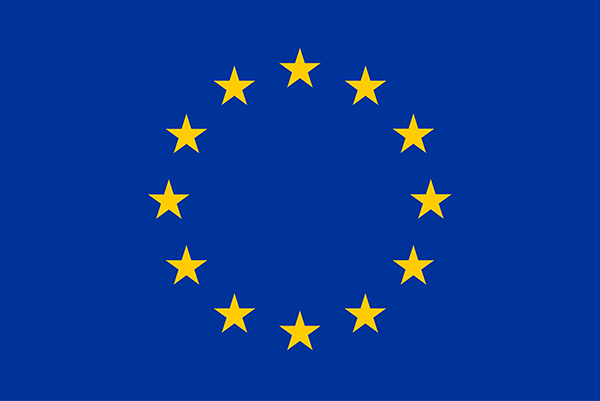}


\bibliographystyle{abbrvnat}
\bibliography{hindex}

\end{document}